\documentclass[%
 reprint,
superscriptaddress,
groupedaddress,
nofootinbib,
amsmath,amssymb,
aps,
prl,
floatfix,
]{revtex4-2}
\usepackage{graphicx}
\usepackage{xcolor}
\usepackage{bbold}
\usepackage{amsthm}
\usepackage{dsfont}
\newtheorem{proposition}{Proposition}
\newtheorem{corollary}{Corollary}

\usepackage{physics}

\definecolor{mymagenta}{RGB}{200, 0, 100}
\definecolor{myblue}{RGB}{45, 48, 146}

\usepackage[pdftex,
            pdftitle={Imaginarity as a necessary resource for trainability in QAOA},
            pdfauthor={Kostas Blekos, Syed Muhammad Ali Hassan, Stefan Kuehn, Nikos Kollas, Karl Jansen},
            bookmarks,
            colorlinks,
            linkcolor=myblue,
            citecolor=mymagenta,
            menucolor=black,
            urlcolor=myblue,
            plainpages=false,
            pdfpagelabels,
            hypertexnames=false]{hyperref}
\usepackage{orcidlink}
\begin{document}

\title{Imaginarity as a necessary resource for trainability in QAOA}

\author{Syed Muhammad Ali Hassan\,\orcidlink{0000-0001-9132-2846}}
\email{s.m.hassan@cyi.ac.cy}
\affiliation{Computation-based Science and Technology Research Center, The Cyprus Institute, Nicosia, Cyprus}

\author{Kostas Blekos\,\orcidlink{0000-0002-6777-2107}}
\affiliation{Computation-based Science and Technology Research Center, The Cyprus Institute, Nicosia, Cyprus}

\author{Stefan K{\"u}hn\,\orcidlink{0000-0001-7693-350X}}
\affiliation{Deutsches Elektronen-Synchrotron DESY, Platanenallee 6, 15738 Zeuthen, Germany}

\author{Nikos Kollas\,\orcidlink{0000-0003-0763-1211}}
\affiliation{Department of Physics, University of Patras, Patras, Greece}

\author{Karl Jansen}
\affiliation{Computation-based Science and Technology Research Center, The Cyprus Institute, Nicosia, Cyprus}
\affiliation{Deutsches Elektronen-Synchrotron DESY, Platanenallee 6, 15738 Zeuthen, Germany}

\date{\today}

\begin{abstract}
The quantum approximate optimization algorithm (QAOA) tackles combinatorial problems by tuning a quantum circuit in a classical loop, often guided by gradients.
We show that the gradient used to tune the circuit's final parameter is bounded by imaginarity, which weights phase relationships between candidate solutions by how strongly the circuit connects them and how differently the problem scores them.
Imaginarity is necessary but not sufficient for a nonzero gradient.
We extend the bound to three common noise models and compare it numerically with the gradient in Max-Cut simulations.
\end{abstract}

\maketitle


\section{Introduction}

The quantum approximate optimization algorithm (QAOA) is a hybrid quantum--classical algorithm for combinatorial optimization: it searches a large set of candidate solutions for one that maximizes an objective function~\cite{farhi2014quantumapproximateoptimizationalgorithm}.
It is often studied as a candidate for near-term, noisy quantum hardware, since it uses shallow, problem-structured circuits~\cite{Preskill_2018,Abbas2024}.
The circuit alternates problem operations, which encode the objective function, with mixer operations, which redistribute probability among candidate solutions.
Each operation is applied with an adjustable angle, and a classical optimizer tunes these angles to maximize the expected objective value.
Training typically relies on the gradient of this expectation with respect to an angle, and when that gradient is small, optimization becomes difficult.
In asymptotic settings, this loss of trainable signal is related to barren-plateau phenomena in variational quantum algorithms~\cite{McClean_2018,Cerezo_2021,Arrasmith_2021,Larocca2025}.
An upper bound on a gradient can identify conditions under which a parameter cannot provide a useful update: when the bound vanishes, the corresponding derivative must vanish.
Such a bound can identify state resources required for a gradient signal, but it does not guarantee trainability or a nonzero gradient.

Quantum resources such as coherence and entanglement play a crucial role in many quantum algorithms~\cite{PhysRevA.93.012111,PhysRevA.106.062429,PhysRevLett.129.120501,e21030260,FENG2023129048}.
They also play a role in QAOA: coherence of the input state can guide the choice of initialization~\cite{Caliz2025}, while quantum Fisher information has been used to study entanglement and parameter sensitivity~\cite{Sarmina2026}.
Here we identify cost-basis imaginarity~\cite{Hickey_2018,Wu_2021,Wu_2021a} as a necessary resource for a nonzero gradient with respect to the final mixer angle.

For any diagonal cost Hamiltonian and any mixer that is real in the computational basis, this gradient receives contributions only from imaginary coherences between basis states connected by the mixer.
The mixer selects the relevant off-diagonal entries, while differences in cost weight their contributions.
This gives an exact identity, an upper bound, and a necessary condition for a nonzero gradient; the converse does not follow because different contributions may cancel.
For the standard transverse-field mixer, the selected coherences connect bitstrings that differ by one bit.
For terminal noise, the same identity applies when the cost observable propagated through the adjoint channel remains diagonal, with the bound evaluated on the pre-channel state and the effective cost observable.

We first introduce the QAOA and imaginarity notation, then derive the noiseless and terminal-noise bounds.
We use Max-Cut as a concrete example and compare depolarizing, phase-flip, and amplitude-damping channels in exact simulations of small instances.

\section{Background and setup}

\subsection{\label{sec:qaoa}QAOA}

QAOA is a hybrid quantum-classical algorithm for combinatorial optimization problems.
Its quantum part prepares a parametric ansatz state from alternating cost and mixer layers~\cite{farhi2014quantumapproximateoptimizationalgorithm}.
The relevant objects are the cost Hamiltonian $H_C$, the mixer Hamiltonian $H_B$, and the cost expectation $\Tr(\rho H_C)$.

For each bitstring $\mathbf{z}\in\{0,1\}^n$, let $C_{\mathbf{z}}$ be its real objective value.
The cost Hamiltonian is defined as
\begin{equation}
H_C = \sum_{\mathbf{z}\in\{0,1\}^n} C_{\mathbf{z}}\,\ketbra{\mathbf{z}},
\label{eq:cost_hamiltonian}
\end{equation}
where $\ket{\mathbf{z}}$ is the corresponding computational-basis state.

The standard mixer (used later in this work) is the \emph{transverse-field Hamiltonian}
\begin{equation}
H_B = \sum_{i=1}^{n} X_i.
\label{eq:mixer}
\end{equation}
Each $X_i$ flips one qubit, so $H_B$ connects computational-basis states that differ in one bit.

The algorithm starts from the uniform superposition of computational states
\begin{equation}
\ket{s}=\ket{+}^{\otimes n}
=\frac{1}{\sqrt{2^n}}\sum_{\mathbf{z}\in\{0,1\}^n}\ket{\mathbf{z}}.
\end{equation}
A depth-$p$ circuit applies $p$ alternating cost and mixer layers to this state, producing
\begin{equation}
\ket{\psi_p(\boldsymbol{\gamma},\boldsymbol{\beta})}
=
e^{-i\beta_p H_B}e^{-i\gamma_p H_C}\cdots
e^{-i\beta_1 H_B}e^{-i\gamma_1 H_C}\ket{s}.
\end{equation}
Each layer has one cost angle and one mixer angle.
The variational parameters are the cost angles $\boldsymbol{\gamma}=(\gamma_1,\ldots,\gamma_p)$ and mixer angles $\boldsymbol{\beta}=(\beta_1,\ldots,\beta_p)$ for all layers.
A classical optimizer adjusts these angles.
During training, the goal is to maximize the cost expectation
\begin{equation}
F_p(\boldsymbol{\gamma},\boldsymbol{\beta})
=
\bra{\psi_p(\boldsymbol{\gamma},\boldsymbol{\beta})}
H_C
\ket{\psi_p(\boldsymbol{\gamma},\boldsymbol{\beta})}.
\end{equation}

\subsection{\label{sec:cost-basis-imaginarity}Imaginarity}

Throughout this work, the reference basis is the computational basis: each basis state represents a candidate solution, and the cost Hamiltonian is diagonal in this basis.
The off-diagonal matrix elements describe coherences between different candidate solutions, and imaginarity refers to the imaginary parts of these coherences.
The resource theory of imaginarity defines real density matrices as free states and treats non-real matrix entries in a fixed basis as resourceful~\cite{Hickey_2018}.
Consider a density operator $\rho$ on a $d$-dimensional Hilbert space $\mathcal{H}$ with orthonormal reference basis
$\{\ket{z}\}_{z=0}^{d-1}$,

\begin{equation}
\rho
=
\sum_{z,z'=0}^{d-1}\rho_{zz'}\ketbra{z}{z'}.
\end{equation}

We say that $\rho$ is real, or free, if $\rho_{zz'}\in\mathds{R}$ for all $z,z'$.
States with nonzero imaginary off-diagonal elements are resourceful.
Imaginarity is useful for quantum information tasks including local state discrimination and state conversion~\cite{Wu_2021,Wu_2021a}.
It has also been studied in pseudorandomness, multiparameter estimation, quantum circuit complexity, and quantum speed limits~\cite{haug2025pseudorandom,miyazaki2022imaginarity,ye2026coherence,s7kr-8rrn}.
The sum of the absolute values of the imaginary parts of the off-diagonal elements defines a measure of imaginarity,
\begin{equation}
I_{\ell_1}(\rho)
=
\sum_{z\ne z'}\left|\operatorname{Im}\left(\rho_{zz'}\right)\right|,
\end{equation}
known as the $\ell_1$-norm of imaginarity~\cite{Chen_2023}.

\section{Mixer-gradient bounds}

In the previous section we defined imaginarity as a basis-dependent property of a quantum state.
We now connect it to the QAOA gradient with respect to the final mixer angle.
The calculation requires a diagonal cost Hamiltonian, a mixer that is real in the same basis, and an input state that is independent of that mixer angle.

\subsection{Noiseless QAOA}

Let $H_C$ be the diagonal cost Hamiltonian in Eq.~\eqref{eq:cost_hamiltonian},
and let $H_B$ be a Hermitian mixer whose computational-basis matrix elements are real.
For
\begin{equation}
\rho_\beta=e^{-i\beta H_B}\rho_0e^{i\beta H_B},
\qquad
F(\beta)=\Tr(H_C\rho_\beta),
\end{equation}
assume that $\rho_0$ is independent of $\beta$.
For QAOA, $\beta$ can be taken to be the final mixer angle $\beta_p$, with $\rho_0$ denoting the state immediately before that final mixer.

We define the ordered mixer support
\begin{equation}
\mathcal{S}_B
=
\{(\mathbf{z},\mathbf{z}'):\mathbf{z}\ne\mathbf{z}',\ (H_B)_{\mathbf{z}\mathbf{z}'}\ne0\}.
\end{equation}
Pairs in $\mathcal{S}_B$ are called mixer edges.
The computational-basis states can therefore be viewed as vertices of a graph, with a mixer edge joining two candidate solutions that the mixer directly couples.
This mixer graph is distinct from the problem graph used in the Max-Cut example below.
All sums over $\mathcal{S}_B$ below are ordered-pair sums.

\begin{proposition}[Mixer-gradient identity]
\label{prop:mixer-edge-gradient}
For a diagonal cost Hamiltonian and a real mixer,
\begin{equation}
\partial_\beta F
=
-\sum_{(\mathbf{z},\mathbf{z}')\in\mathcal{S}_B}
(C_{\mathbf{z}}-C_{\mathbf{z}'})
(H_B)_{\mathbf{z}\mathbf{z}'}
\operatorname{Im}(\rho_\beta)_{\mathbf{z}\mathbf{z}'}.
\label{eq:mixer-gradient-identity}
\end{equation}
Consequently,
\begin{equation}
|\partial_\beta F|
\le
\sum_{(\mathbf{z},\mathbf{z}')\in\mathcal{S}_B}
|C_{\mathbf{z}}-C_{\mathbf{z}'}|
|(H_B)_{\mathbf{z}\mathbf{z}'}|
|\operatorname{Im}(\rho_\beta)_{\mathbf{z}\mathbf{z}'}|.
\label{eq:mixer-gradient-weighted-bound}
\end{equation}
\end{proposition}

In simpler terms, a pair of candidate solutions can contribute to the derivative only if the mixer connects them, their objective values differ, and their coherence has a nonzero imaginary part.
If all such imaginary coherences vanish, then the derivative vanishes; the converse need not hold because different contributions can cancel.

We also define the following quantities, which will be used later:
\begin{align}
\Delta_{\mathcal{S}_B}(H_C)
&=
\max_{(\mathbf{z},\mathbf{z}')\in\mathcal{S}_B}
|C_{\mathbf{z}}-C_{\mathbf{z}'}|,
\\
M_{\mathcal{S}_B}(H_B)
&=
\max_{(\mathbf{z},\mathbf{z}')\in\mathcal{S}_B}
|(H_B)_{\mathbf{z}\mathbf{z}'}|,
\\
I_{\ell_1}^{(\mathcal{S}_B)}(\rho)
&=
\sum_{(\mathbf{z},\mathbf{z}')\in\mathcal{S}_B}
|\operatorname{Im}\rho_{\mathbf{z}\mathbf{z}'}|.
\end{align}
The last quantity is the mixer-edge imaginarity, namely the imaginarity supported on pairs directly connected by the mixer.

\begin{corollary}[Mixer-support and coarse bounds]
\label{cor:mixer-support-bound}
Then
\begin{equation}
|\partial_\beta F|
\le
\Delta_{\mathcal{S}_B}(H_C)M_{\mathcal{S}_B}(H_B)
I_{\ell_1}^{(\mathcal{S}_B)}(\rho_\beta).
\end{equation}
\end{corollary}

The proof of Proposition~\ref{prop:mixer-edge-gradient} and Corollary~\ref{cor:mixer-support-bound} is given in the Appendix.

\subsection{Terminal noise and the adjoint observable}

Let a channel $\mathcal{E}_\lambda$, where $\lambda$ denotes the noise strength, act after the variational circuit, so that
\begin{equation}
\sigma_\lambda(\beta)=\mathcal{E}_\lambda(\rho_\beta),
\qquad
F_\lambda(\beta)=\Tr(H_C\sigma_\lambda(\beta)).
\end{equation}
The adjoint channel is defined by
\begin{equation}
\Tr[A\mathcal{E}_\lambda(B)]
=
\Tr[\mathcal{E}_\lambda^\dagger(A)B].
\end{equation}
We can express the noisy expectation using an effective objective observable evaluated on the state before noise:
\begin{equation}
F_\lambda(\beta)
=
\Tr(H_C^{(\lambda)}\rho_\beta),
\qquad
H_C^{(\lambda)}=\mathcal{E}_\lambda^\dagger(H_C).
\end{equation}
The two expressions are equivalent: the first applies noise to the state, while the second incorporates the same noise into $H_C^{(\lambda)}$ through the adjoint channel.

\begin{proposition}[Terminal-noise adjoint-channel bound]
\label{prop:terminal-adj-bound}
When the adjoint-propagated observable remains diagonal in the computational basis,
\begin{equation}
H_C^{(\lambda)}
=
\sum_{\mathbf{z}}C_{\mathbf{z}}^{(\lambda)}\ketbra{\mathbf{z}}.
\end{equation}
Then
\begin{align}
|\partial_\beta F_\lambda|
&\le
B_{\rm adj}^{(\lambda)}(\rho_\beta),
\label{eq:terminal-adj-bound}
\\
B_{\rm adj}^{(\lambda)}(\rho)
&=
\sum_{(\mathbf{z},\mathbf{z}')\in\mathcal{S}_B}
|C_{\mathbf{z}}^{(\lambda)}-C_{\mathbf{z}'}^{(\lambda)}|
|(H_B)_{\mathbf{z}\mathbf{z}'}|
|\operatorname{Im}\rho_{\mathbf{z}\mathbf{z}'}|
\\
&\le
\Delta_{\mathcal{S}_B}(H_C^{(\lambda)})M_{\mathcal{S}_B}(H_B)
I_{\ell_1}^{(\mathcal{S}_B)}(\rho).
\end{align}
\end{proposition}

We call $B_{\rm adj}^{(\lambda)}(\rho)$ the weighted adjoint-channel bound.
It combines the mixer-edge imaginarity of the state before terminal noise with the noise-modified cost differences of $H_C^{(\lambda)}$.
For $\lambda=0$, the channels used below reduce to the identity, and $B_{\rm adj}^{(0)}(\rho_\beta)$ reduces to the noiseless weighted bound in Eq.~\eqref{eq:mixer-gradient-weighted-bound}.

This statement needs only the diagonality of $\mathcal{E}_\lambda^\dagger(H_C)$, not self-adjointness of the channel or a resource-theoretic classification of its Kraus operators.
It therefore applies to all terminal channels used below.

\subsection{Transverse-field mixer}

Let $d_H(\mathbf{z},\mathbf{z}')$ denote the Hamming distance between two bitstrings, namely the number of bit positions at which they differ.
For the transverse-field mixer in Eq.~\eqref{eq:mixer},
$(H_B)_{\mathbf{z}\mathbf{z}'}=1$ exactly when $d_H(\mathbf{z},\mathbf{z}')=1$.
Thus $\mathcal{S}_B$ is the ordered edge set of the Hamming hypercube: its vertices are all bitstrings, and its edges connect pairs related by a one-bit flip.
We define
\begin{align}
\Delta_1(H_C)
&=
\max_{d_H(\mathbf{z},\mathbf{z}')=1}
|C_{\mathbf{z}}-C_{\mathbf{z}'}|,
\\
I_{\ell_1}^{(1)}(\rho)
&=
\sum_{d_H(\mathbf{z},\mathbf{z}')=1}
|\operatorname{Im}\rho_{\mathbf{z}\mathbf{z}'}|.
\label{eq:restricted-imaginarity}
\end{align}
The noiseless transverse-field bound is therefore
\begin{equation}
|\partial_\beta F|
\le
\Delta_1(H_C)I_{\ell_1}^{(1)}(\rho_\beta).
\label{eq:transverse-field-bound}
\end{equation}
If $H_C^{(\lambda)}$ remains diagonal, the terminal-noise specialization is
\begin{align}
|\partial_\beta F_\lambda|
&\le
B_{\rm adj}^{(\lambda)}(\rho_\beta)
\\
&=
\sum_{d_H(\mathbf{z},\mathbf{z}')=1}
|C_{\mathbf{z}}^{(\lambda)}-C_{\mathbf{z}'}^{(\lambda)}|
|\operatorname{Im}(\rho_\beta)_{\mathbf{z}\mathbf{z}'}|
\\
&\le
\Delta_1(H_C^{(\lambda)})I_{\ell_1}^{(1)}(\rho_\beta).
\end{align}

\section{A concrete example: the Max-Cut problem}

Max-Cut asks for a partition of the vertices of a graph into two sets that maximizes the number of edges joining different sets.
It is a canonical NP-hard graph-partitioning problem with applications in circuit-layout design, statistical physics, network science, and clustering.
Its direct mapping to a diagonal Ising Hamiltonian also makes it a standard QAOA benchmark~\cite{farhi2014quantumapproximateoptimizationalgorithm,chang2023quantumspeedupmaximumcut,Zhou_2020,Blekos_2024}.
For a graph $G=(V,E_G)$ with vertex set $V$ and edge set $E_G$, a bitstring $\mathbf{z}\in\{0,1\}^{|V|}$ assigns every vertex to one of two subsets.
An edge $(i,j)$ crosses the partition when $z_i\ne z_j$, and the unweighted cut value is
\begin{equation}
C_{\mathbf{z}}
=
\sum_{(i,j)\in E_G}\mathbf{1}_{\{z_i\ne z_j\}}.
\end{equation}
Here $\mathbf{1}_{\{z_i\ne z_j\}}$ equals one when the endpoints belong to different sets and zero otherwise, so $C_{\mathbf{z}}$ counts the graph edges crossing the partition.
Max-Cut maximizes $C_{\mathbf{z}}$, with corresponding cost Hamiltonian
\begin{equation}
    H_C
    =
    \frac{1}{2}
    \sum_{(i,j)\in E_G}
    \left(\mathbb{1}-Z_iZ_j\right).
\end{equation}
The transverse-field mixer connects each bitstring to bitstrings obtained by flipping one vertex assignment.
Only graph edges incident to that vertex can change their crossing status.
Let $\Delta_{\max}$ be the maximum graph degree, namely the largest number of graph edges incident on any vertex.
Then $\Delta_1(H_C)\le\Delta_{\max}$.

\begin{corollary}[Transverse-field Max-Cut bound]
\label{cor:max-cut-bound}
For unweighted Max-Cut on a graph with maximum degree $\Delta_{\max}$ and the transverse-field mixer,
\begin{equation}
\abs{\partial_\beta F}
\le
\Delta_{\max}I_{\ell_1}^{(1)}(\rho_\beta).
\end{equation}
\end{corollary}
The proof is given in the Appendix.
This is the graph-local version of the transverse-field bound.
For the numerical setting below, it explains why Hamming-one imaginarity is the quantity selected directly by the mixer.

For the noisy bound, the terminal channel is applied independently to each qubit.
At the one-qubit level, phase-flip noise applies $Z$ with probability $\lambda$, depolarizing noise applies each Pauli error $X$, $Y$, and $Z$ with probability $\lambda/3$, and amplitude damping models decay from $\ket{1}$ to $\ket{0}$ with probability $\lambda$.
The corresponding channels are
\begin{align}
\mathcal{P}_\lambda(\tau)
&=
(1-\lambda)\tau+\lambda Z\tau Z,
\\
\mathcal{D}_\lambda(\tau)
&=
(1-\lambda)\tau
+
\frac{\lambda}{3}\left(X\tau X+Y\tau Y+Z\tau Z\right),
\\
\mathcal{A}_\lambda(\tau)
&=
K_0\tau K_0^\dagger+K_1\tau K_1^\dagger,
\end{align}
where
\begin{equation}
K_0=\ketbra{0}+\sqrt{1-\lambda}\ketbra{1},
\qquad
K_1=\sqrt{\lambda}\ketbra{0}{1}.
\end{equation}
The corresponding one-qubit adjoint actions on $Z$ are
\begin{align}
    \mathcal{P}_\lambda^\dagger(Z)
    &=
    Z,
    \\
    \mathcal{D}_\lambda^\dagger(Z)
    &=
    \left(1-\frac{4\lambda}{3}\right)Z,
    \\
    \mathcal{A}_\lambda^\dagger(Z)
    &=
    (1-\lambda)Z+\lambda\mathbb{1}.
\end{align}
For tensor-product noise, these identities apply independently to each Pauli factor in $H_C$.
Thus phase-flip noise leaves the Max-Cut observable unchanged in the adjoint picture.
Depolarizing noise rescales each $Z_iZ_j$ term by $\left(1-4\lambda/3\right)^2$, while
amplitude damping maps $Z_iZ_j$ to $\big((1-\lambda)Z_i+\lambda\mathbb{1}\big)\big((1-\lambda)Z_j+\lambda\mathbb{1}\big)$.
All three adjoint observables remain diagonal in the computational basis, so the terminal-noise Hamming-one bound applies with the restricted spread $\Delta_1(H_C^{(\lambda)})$.
The amplitude-damping adjoint observable contains two-body $Z_iZ_j$ terms, one-body $Z_i$ terms, and constants.

\begin{corollary}[Noisy transverse-field Max-Cut bound]
\label{cor:noisy-max-cut-bound}
For independent phase-flip, depolarizing, or amplitude-damping noise with $0\le\lambda\le1$,
\begin{equation}
\Delta_1(H_C^{(\lambda)})\le\Delta_{\max},
\qquad
|\partial_\beta F_\lambda|
\le
\Delta_{\max}I_{\ell_1}^{(1)}(\rho_\beta).
\end{equation}
\end{corollary}
The proof is given in the Appendix.
The guaranteed weighted terminal-noise bound remains $B_{\rm adj}^{(\lambda)}(\rho_\beta)$, which retains the individual channel-modified cost differences rather than replacing them by their maximum.

For tensor-product phase-flip noise, this difference can be written explicitly.
For computational-basis strings $\mathbf{z}$ and $\mathbf{z}'$,
\begin{equation}
\mathcal{P}_\lambda^{\otimes n}
\!\left(\ket{\mathbf{z}}\bra{\mathbf{z}'}\right)
=
(1-2\lambda)^{d_H(\mathbf{z},\mathbf{z}')}
\ket{\mathbf{z}}\bra{\mathbf{z}'}.
\end{equation}
It follows that
$I_{\ell_1}^{(1)}(\sigma_\lambda)=|1-2\lambda|I_{\ell_1}^{(1)}(\rho_\beta)$.
At the same time, $\mathcal{P}_\lambda^\dagger(H_C)=H_C$, so
$F_\lambda(\beta)=F(\beta)$ and $\partial_\beta F_\lambda=\partial_\beta F$.
Terminal phase-flip noise can therefore reduce the final-state imaginarity while leaving the gradient and its adjoint-channel bound unchanged.
At $\lambda=1/2$, the final state is fully dephased and its imaginarity vanishes, although the gradient can remain nonzero.

\subsection{\label{sec:numerical-analysis}Numerical analysis}

The figures below use exact dense-state simulation of depth-$p=1$ QAOA for Max-Cut on random regular graphs.
Noise is applied once, as a terminal channel after the variational circuit.
The first figure compares the weighted adjoint-channel bound $B_{\rm adj}^{(\lambda)}(\rho_\beta)$ pointwise with $|\partial_\beta F_\lambda|$.
Whenever $B_{\rm adj}^{(\lambda)}(\rho_\beta)>0$, the saturation ratio
$R_{\rm adj}=|\partial_\beta F_\lambda|/B_{\rm adj}^{(\lambda)}(\rho_\beta)$ is at most one; the second figure shows how its distribution changes with noise strength.
For comparison with ordinary coherence in the fourth figure, we use
\begin{equation}
C_{\ell_1}^{(1)}(\rho)
=
\sum_{d_H(\mathbf{z},\mathbf{z}')=1}
|\rho_{\mathbf{z}\mathbf{z}'}|.
\end{equation}
This quantity includes real as well as imaginary off-diagonal components; the mixer-gradient identity depends only on the imaginary components.
The main ensemble uses $n=6$, graph degrees $k\in\{3,4,5\}$, and graph counts $2$, $2$ and $1$, respectively.
A $k$-regular graph is one in which every vertex has degree $k$.
For each graph, we sample $18$ parameter pairs independently with $\beta\in[0,\pi/2]$ and $\gamma\in[0,\pi]$.
The terminal noise strengths are $\lambda\in\{0,0.06,0.12,0.18,0.24,0.30\}$.
The finite-size panel uses one $3$-regular graph for each $n\in\{4,6,8\}$, with the same number of parameter samples and fixed $\lambda=0.16$.
Duplicate graph edge sets are rejected, and graphs are retained regardless of connectedness.
Gradients are evaluated analytically from the adjoint-channel expression.
All Max-Cut costs and gradients are unnormalized, and all Hamming-one coherence sums use the ordered-pair convention.
Medians and percentile bands are taken over the combined graph-parameter ensemble for the corresponding panel.

\begin{figure*}[t]
    \centering
    \includegraphics[width=\textwidth]{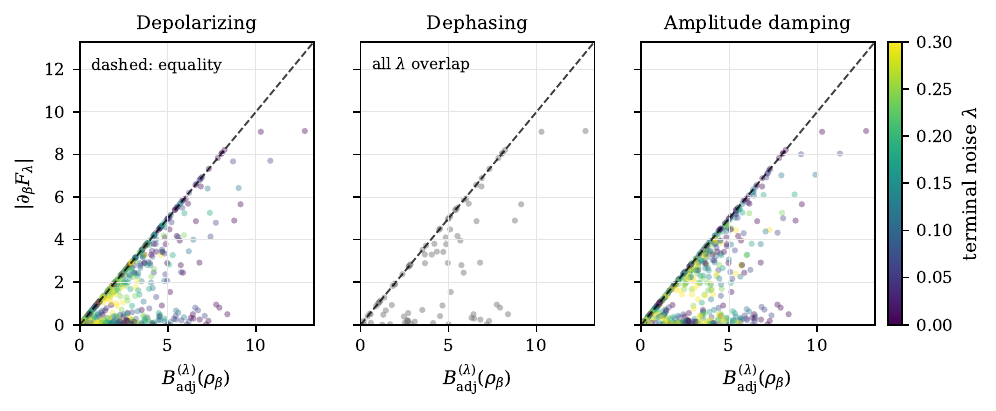}
    \caption{\label{fig:adjoint-bound-by-channel}
    Per-instance test of the terminal-noise adjoint-channel bound.
    Each point is one depth-$p=1$ QAOA instance on an $n=6$ random $k$-regular Max-Cut graph with $k\in\{3,4,5\}$ and one parameter pair $(\beta,\gamma)$.
    The horizontal axis is the weighted bound $B_{\rm adj}^{(\lambda)}(\rho_\beta)$ formed from the pre-channel state and the adjoint-propagated cost observable; the vertical axis is $|\partial_\beta F_\lambda|$.
    Color denotes the terminal noise strength $\lambda$.
    The dashed diagonal marks saturation, and all points lie on or below it.
    Distance below the diagonal results from cancellations between signed mixer-edge contributions.
    For phase-flip noise, both axes are independent of $\lambda$, so the coincident points for the six noise strengths are plotted once in gray.}
\end{figure*}

\begin{figure}[t]
    \centering
    \includegraphics[width=\columnwidth]{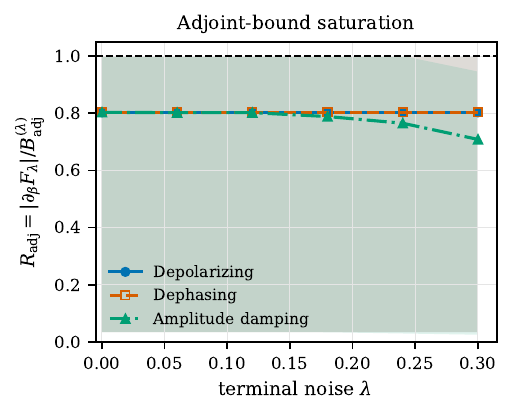}
    \caption{\label{fig:saturation-ratio-vs-noise}
    Saturation of the rigorous adjoint-channel bound under terminal noise.
    The ensemble is the same as in Fig.~\ref{fig:adjoint-bound-by-channel}; curves show medians and shaded regions show 10th--90th percentiles.
    The dashed line at one is the rigorous upper limit whenever the denominator is nonzero.
    For depolarizing and phase-flip noise, the gradient and weighted bound have the same channel dependence, so $R_{\rm adj}$ is independent of $\lambda$; the corresponding curves overlap.
    Amplitude damping changes the individual cost-difference weights nonuniformly, and the median saturation decreases at larger $\lambda$.}
\end{figure}

\begin{figure*}[t]
    \centering
    \includegraphics[width=\textwidth]{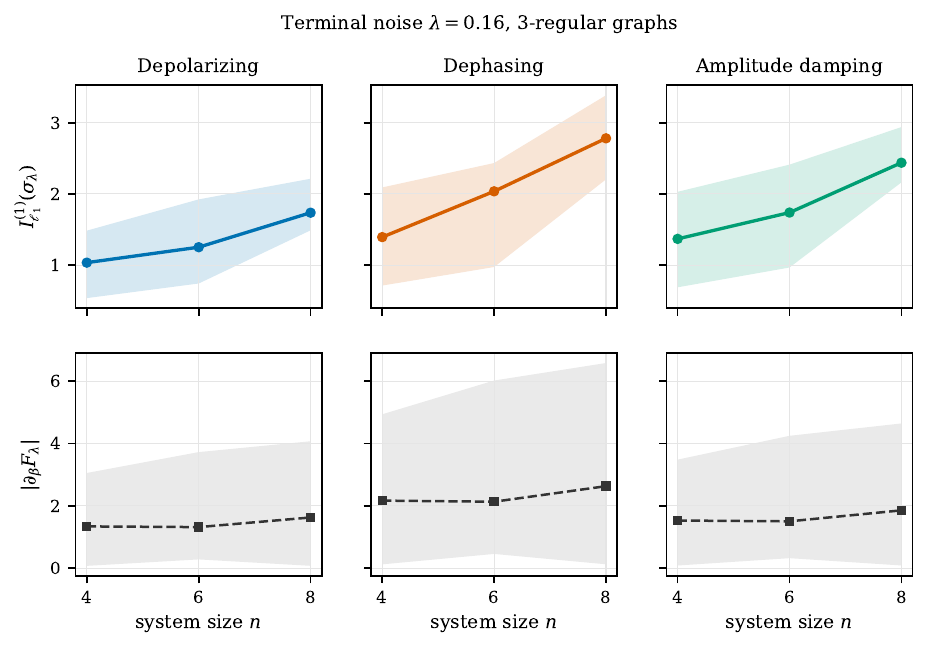}
    \caption{\label{fig:coherence-gradient-scaling}
    Finite-size tracking of Hamming-one imaginarity and mixer-gradient magnitude at fixed terminal noise strength.
    The top row shows the ensemble median of $I_{\ell_1}^{(1)}(\sigma_\lambda)$ for $3$-regular Max-Cut instances with $n\in\{4,6,8\}$ and $\lambda=0.16$.
    The bottom row shows the corresponding median of $|\partial_\beta F_\lambda|$.
    Shaded regions denote 10th--90th percentile bands over sampled graphs and parameters.
    This panel summarizes the finite samples; the adjoint-channel bound is tested in Fig.~\ref{fig:adjoint-bound-by-channel}, and no asymptotic conclusion is drawn from the small ensemble.}
\end{figure*}

\begin{figure*}[t]
    \centering
    \includegraphics[width=\textwidth]{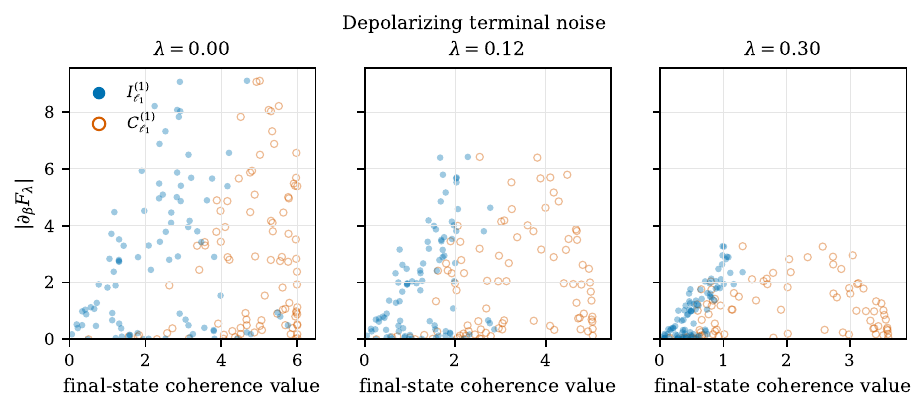}
    \caption{\label{fig:im-l1-depolarizing-comparison}
    Comparison of Hamming-one imaginarity and restricted $\ell_1$ coherence on the same terminal noisy states.
    The panels show depolarizing terminal noise at low, intermediate and high $\lambda$.
    Filled blue points use $I_{\ell_1}^{(1)}(\sigma_\lambda)$ on the horizontal axis.
    Open orange points use $C_{\ell_1}^{(1)}(\sigma_\lambda)$ for the same states and gradients.
    Since $I_{\ell_1}^{(1)}\le C_{\ell_1}^{(1)}$ pointwise, the $\ell_1$ points shift to the right for the same gradient values.
    Both horizontal quantities are computed from the final noisy state.
    The rigorous bound, which uses the pre-channel state and adjoint observable, is shown in Fig.~\ref{fig:adjoint-bound-by-channel}.}
\end{figure*}

\section{Discussion}
Changing the final mixer angle can guide QAOA training only if it changes the expected objective value.
Our result shows that this change can arise only from imaginary coherences between candidate solutions directly coupled by the mixer.
For the standard transverse-field mixer, these are pairs of bitstrings that differ by one bit.
Mixer-edge imaginarity is therefore necessary for a nonzero final mixer-angle gradient, but cancellations mean that it does not by itself ensure successful optimization.

Quantitatively, the gradient is limited by the mixer-edge imaginarity and the cost differences across those edges.
The restricted $\ell_1$ coherence gives a larger bound because it also counts real off-diagonal components.
Extending the argument to relative entropy, robustness, or other nonlinear coherence measures would require a separate proof.

The result is a one-sided limitation on one gradient signal in this QAOA setting.
A state may have nonzero Hamming-one imaginarity, while cancellations, graph structure, or the chosen parameters can still make the derivative small or zero.
Approximation guarantees and successful optimization require additional arguments.

The noise analysis here uses a terminal channel so that the noisy final state and the noisy gradient can be separated explicitly.
In this model, the gradient can be written with the pre-channel state and the cost Hamiltonian propagated through the adjoint channel.
When this adjoint-channel observable remains diagonal, the Hamming-distance-one calculation gives the weighted bound $B_{\rm adj}^{(\lambda)}(\rho_\beta)$ with channel-modified cost differences.
The final noisy state and the adjoint-channel observable can give different numerical quantities for the same derivative.
For phase-flip noise, this distinction is explicit: at $\lambda=1/2$, the channel completely dephases the final state and removes its imaginarity, while the diagonal objective and its gradient remain unchanged.

In the finite instances studied here, the gradient obeys the weighted adjoint-channel bound pointwise.
The final-state imaginarity alone is insufficient to determine the channel dependence of the gradient or to test the bound.
Restricted $\ell_1$ coherence is less specific because it counts real as well as imaginary off-diagonal components.
These plots are small exact-simulation checks of the analytic relations.
Studying asymptotic scaling would require a separate finite-size analysis.

Several extensions are still open.
The final mixer angle at higher depth follows from the same calculation.
Earlier mixer angles require the cost observable to be propagated through later layers, since QAOA then has several mixer angles and intermediate states.
Other mixers would replace the Hamming-distance-one restriction by the connectivity induced by the mixer.
For other diagonal objectives, the graph-degree factor would be replaced by the largest objective change across one allowed mixer transition.
These extensions would test how much of the present mixer-edge imaginarity statement extends beyond depth-one Max-Cut.
In the setting studied here, mixer connectivity provides a concrete selection rule for the imaginary coherences that can support the final mixer-gradient.

\appendix

\section{Proofs of the gradient bounds}

\subsection{Proof of Proposition~\ref{prop:mixer-edge-gradient}}

\begin{proof}
The product rule gives
\begin{equation}
\partial_\beta \rho_\beta
=
-iH_B\rho_\beta+i\rho_\beta H_B
=
-i[H_B,\rho_\beta].
\end{equation}
Therefore,
\begin{align}
\frac{\partial F}{\partial \beta}
&=
\Tr(H_C\,\partial_\beta \rho_\beta) \\
&=
-i\,\Tr\!\left(H_C[H_B,\rho_\beta]\right).
\end{align}
Cyclicity of the trace gives
\begin{equation}
\Tr\!\left(H_C[H_B,\rho_\beta]\right)
=
\Tr\!\left([H_C,H_B]\rho_\beta\right).
\end{equation}
The commutator has matrix elements
\begin{equation}
[H_C,H_B]_{\mathbf{z}\mathbf{z}'}
=
(C_{\mathbf{z}} - C_{\mathbf{z}'})(H_B)_{\mathbf{z}\mathbf{z}'}.
\end{equation}
Because $H_B$ is real and Hermitian, the commutator matrix is real and antisymmetric.
Combining the last two equations gives
\begin{equation}
\frac{\partial F}{\partial\beta}
=
-i
\sum_{\mathbf{z},\mathbf{z}'}
(C_{\mathbf{z}} - C_{\mathbf{z}'})
(H_B)_{\mathbf{z}\mathbf{z}'}
(\rho_\beta)_{\mathbf{z}'\mathbf{z}}.
\end{equation}
Pairing the ordered terms $(\mathbf{z},\mathbf{z}')$ and $(\mathbf{z}',\mathbf{z})$ gives
\begin{equation}
\frac{\partial F}{\partial\beta}
=
-\sum_{(\mathbf{z},\mathbf{z}')\in\mathcal{S}_B}
(C_{\mathbf{z}}-C_{\mathbf{z}'})
(H_B)_{\mathbf{z}\mathbf{z}'}
\operatorname{Im}(\rho_\beta)_{\mathbf{z}\mathbf{z}'}.
\end{equation}
Terms outside $\mathcal{S}_B$ vanish by definition, and the triangle inequality gives Eq.~\eqref{eq:mixer-gradient-weighted-bound}.
\end{proof}

\subsection{Proof of Corollary~\ref{cor:mixer-support-bound}}

\begin{proof}
Replacing the cost differences and mixer matrix elements in Eq.~\eqref{eq:mixer-gradient-weighted-bound} by their largest values on $\mathcal{S}_B$ gives the result.
\end{proof}

\subsection{Proof of Corollary~\ref{cor:max-cut-bound}}

\begin{proof}
For the transverse-field mixer $H_B=\sum_{k=1}^n X_k$, the matrix element $(H_B)_{\mathbf{z}\mathbf{z}'}$ equals $1$ exactly when $d_H(\mathbf{z},\mathbf{z}')=1$.
Proposition~\ref{prop:mixer-edge-gradient} therefore gives the bound with $\Delta_1(H_C)$, where $\Delta_1(H_C)$ is the largest one-bit objective-value change.
For unweighted Max-Cut, let $G=(V,E_G)$ be a graph with maximum degree $\Delta_{\max}$ and let $C_{\mathbf{z}}$ be the cut value of bitstring $\mathbf{z}$.
If $\mathbf{z}'=\mathbf{z}\oplus e_k$, only edges incident to vertex $k$ can change their cut status.
Each such edge changes the cut value by at most $1$, so $\Delta_1(H_C)\le\Delta_{\max}$.
Substituting this inequality into the transverse-field bound gives the result.
\end{proof}

\subsection{Proof of Corollary~\ref{cor:noisy-max-cut-bound}}

\begin{proof}
Phase-flip noise leaves $H_C$ unchanged, so $\Delta_1(H_C^{(\lambda)})\le\Delta_{\max}$ follows from Corollary~\ref{cor:max-cut-bound}.

For depolarizing noise, let $a=1-4\lambda/3$.
Its adjoint action gives
\begin{equation}
H_C^{(\lambda)}
=
a^2H_C+\frac{|E_G|}{2}(1-a^2)\mathbb{1}.
\end{equation}
The identity term does not affect cost differences, and $a^2\le1$ for $0\le\lambda\le1$.
Hence
$\Delta_1(H_C^{(\lambda)})=a^2\Delta_1(H_C)\le\Delta_{\max}$.

For amplitude damping, write $s_i=(-1)^{z_i}$ and
$q_i=\lambda+(1-\lambda)s_i$.
The eigenvalue of the adjoint-propagated contribution from an edge $(i,j)$ is
$(1-q_iq_j)/2$.
Flipping bit $k$ changes the contribution of an incident edge $(k,j)$ by an amount whose absolute value is
$(1-\lambda)|q_j|\le1-\lambda$.
Only the edges incident on $k$ can change, and therefore
\begin{equation}
\Delta_1(H_C^{(\lambda)})
\le
(1-\lambda)\Delta_{\max}
\le
\Delta_{\max}.
\end{equation}
The terminal-noise Hamming-one bound then gives the stated gradient inequality for all three channels.
\end{proof}

\subsection{Proof of Proposition~\ref{prop:terminal-adj-bound}}

\begin{proof}
By the definition of the adjoint channel,
\begin{equation}
F_\lambda(\beta)
=
\Tr\!\left(H_C\mathcal{E}_\lambda(\rho_\beta)\right)
=
\Tr\!\left(\mathcal{E}_\lambda^\dagger(H_C)\rho_\beta\right).
\end{equation}
Thus $F_\lambda(\beta)=\Tr(H_C^{(\lambda)}\rho_\beta)$.
The terminal channel is independent of $\beta$, so the derivative acts only on the pre-channel state $\rho_\beta$.
For the final mixer angle, $\partial_\beta\rho_\beta=-i[H_B,\rho_\beta]$.
Therefore,
\begin{equation}
    \partial_\beta F_\lambda
    =
    -i\,\Tr\!\left(H_C^{(\lambda)}[H_B,\rho_\beta]\right)
    =
    -i\,\Tr\!\left([H_C^{(\lambda)},H_B]\rho_\beta\right).
\end{equation}
If $H_C^{(\lambda)}$ is diagonal in the computational basis, Proposition~\ref{prop:mixer-edge-gradient} applies with $H_C$ replaced by $H_C^{(\lambda)}$.
Taking the triangle inequality gives $B_{\rm adj}^{(\lambda)}(\rho_\beta)$, and replacing its weights by their maxima gives the final inequality in Proposition~\ref{prop:terminal-adj-bound}.
\end{proof}

\begin{acknowledgments}
K.B. and S.M.A.H. acknowledge funding from the European Union's Horizon Europe Framework Programme under the ERA Chair scheme, grant agreement No.~101087126 (QUEST).
S.K.\ and K.J.\ acknowledge support from the Ministry of Science, Research and Culture of the State of Brandenburg within the Centre for Quantum Technologies and Applications (CQTA).
\end{acknowledgments}

\section*{Author Contributions}

K.B.: Conceptualization, Methodology, Software, Formal analysis, Investigation, Visualization, Writing--original draft, Writing--review \& editing, Supervision, and Project administration.
S.M.A.H.: Formal analysis, Investigation, Validation, Visualization, Writing--original draft, and Writing--review \& editing.
S.K.: Methodology, Software, Formal analysis, Validation, Visualization, Supervision, and Writing--review \& editing.
N.K.: Conceptualization, Methodology, Formal analysis, and Writing--review \& editing.
K.J.: Scientific discussion and advice, funding acquisition.

\bibliography{references}

@article{Preskill_2018,
   title={Quantum Computing in the NISQ era and beyond},
   volume={2},
   ISSN={2521-327X},
   url={http://dx.doi.org/10.22331/q-2018-08-06-79},
   DOI={10.22331/q-2018-08-06-79},
   journal={Quantum},
   publisher={Verein zur Forderung des Open Access Publizierens in den Quantenwissenschaften},
   author={Preskill, John},
   year={2018},
   month=aug, pages={79} }

@misc{farhi2014quantumapproximateoptimizationalgorithm,
      title={A Quantum Approximate Optimization Algorithm}, 
      author={Edward Farhi and Jeffrey Goldstone and Sam Gutmann},
      year={2014},
      eprint={1411.4028},
      archivePrefix={arXiv},
      primaryClass={quant-ph},
      url={https://arxiv.org/abs/1411.4028}, 
}

@article{Hickey_2018,
   title={Quantifying the imaginarity of quantum mechanics},
   volume={51},
   ISSN={1751-8121},
   url={http://dx.doi.org/10.1088/1751-8121/aabe9c},
   DOI={10.1088/1751-8121/aabe9c},
   number={41},
   journal={Journal of Physics A: Mathematical and Theoretical},
   publisher={IOP Publishing},
   author={Hickey, Alexander and Gour, Gilad},
   year={2018},
   month=sep, pages={414009} }

@article{Chen_2023,
   title={Measures of imaginarity and quantum state order},
   volume={66},
   pages={280312},
   ISSN={1869-1927},
   url={http://dx.doi.org/10.1007/s11433-023-2126-9},
   DOI={10.1007/s11433-023-2126-9},
   number={8},
   journal={Science China Physics, Mechanics \& Astronomy},
   publisher={Springer Science and Business Media LLC},
   author={Chen, Qiang and Gao, Ting and Yan, Fengli},
   year={2023},
   month=jul }

@article{Wu_2021,
   title={Operational Resource Theory of Imaginarity},
   volume={126},
   pages={090401},
   ISSN={1079-7114},
   url={http://dx.doi.org/10.1103/PhysRevLett.126.090401},
   DOI={10.1103/physrevlett.126.090401},
   number={9},
   journal={Phys. Rev. Lett.},
   publisher={American Physical Society (APS)},
   author={Wu, Kang-Da and Kondra, Tulja Varun and Rana, Swapan and Scandolo, Carlo Maria and Xiang, Guo-Yong and Li, Chuan-Feng and Guo, Guang-Can and Streltsov, Alexander},
   year={2021},
   month=Mar }

@article{Wu_2021a,
   title={Resource theory of imaginarity: Quantification and state conversion},
   volume={103},
   pages = {032401},
   ISSN={2469-9934},
   url={http://dx.doi.org/10.1103/PhysRevA.103.032401},
   DOI={10.1103/physreva.103.032401},
   number={3},
   journal={Phys. Rev. A},
   publisher={American Physical Society (APS)},
   author={Wu, Kang-Da and Kondra, Tulja Varun and Rana, Swapan and Scandolo, Carlo Maria and Xiang, Guo-Yong and Li, Chuan-Feng and Guo, Guang-Can and Streltsov, Alexander},
   year={2021},
   month=Mar }

@article{Zhou_2020,
   title={Quantum Approximate Optimization Algorithm: Performance, Mechanism, and Implementation on Near-Term Devices},
   volume={10},
   pages={021067},
   ISSN={2160-3308},
   url={http://dx.doi.org/10.1103/PhysRevX.10.021067},
   DOI={10.1103/physrevx.10.021067},
   number={2},
   journal={Physical Review X},
   publisher={American Physical Society (APS)},
   author={Zhou, Leo and Wang, Sheng-Tao and Choi, Soonwon and Pichler, Hannes and Lukin, Mikhail D.},
   year={2020},
}

@misc{chang2023quantumspeedupmaximumcut,
      title={Quantum Speedup for the Maximum Cut Problem}, 
      author={Weng-Long Chang and Renata Wong and Wen-Yu Chung and Yu-Hao Chen and Ju-Chin Chen and Athanasios V. Vasilakos},
      year={2023},
      eprint={2305.16644},
      archivePrefix={arXiv},
      primaryClass={quant-ph},
      url={https://arxiv.org/abs/2305.16644}, 
}

@article{McClean_2018,
   title={Barren plateaus in quantum neural network training landscapes},
   volume={9},
   ISSN={2041-1723},
   url={http://dx.doi.org/10.1038/s41467-018-07090-4},
   DOI={10.1038/s41467-018-07090-4},
   number={1},
   pages={4812},
   journal={Nature Communications},
   publisher={Springer Science and Business Media LLC},
   author={McClean, Jarrod R. and Boixo, Sergio and Smelyanskiy, Vadim N. and Babbush, Ryan and Neven, Hartmut},
   year={2018},
   month=nov }

@article{Cerezo_2021,
   title={Cost function dependent barren plateaus in shallow parametrized quantum circuits},
   volume={12},
   pages={1791},
   ISSN={2041-1723},
   url={http://dx.doi.org/10.1038/s41467-021-21728-w},
   DOI={10.1038/s41467-021-21728-w},
   number={1},
   journal={Nature Communications},
   publisher={Springer Science and Business Media LLC},
   author={Cerezo, M. and Sone, Akira and Volkoff, Tyler and Cincio, Lukasz and Coles, Patrick J.},
   year={2021},
   month=mar }

@article{Arrasmith_2021,
   title={Effect of barren plateaus on gradient-free optimization},
   volume={5},
   ISSN={2521-327X},
   url={http://dx.doi.org/10.22331/q-2021-10-05-558},
   DOI={10.22331/q-2021-10-05-558},
   journal={Quantum},
   publisher={Verein zur Forderung des Open Access Publizierens in den Quantenwissenschaften},
   author={Arrasmith, Andrew and Cerezo, M. and Czarnik, Piotr and Cincio, Lukasz and Coles, Patrick J.},
   year={2021},
   month=oct, pages={558} }

@article{Blekos_2024,
   title={A review on Quantum Approximate Optimization Algorithm and its variants},
   volume={1068},
   ISSN={0370-1573},
   url={http://dx.doi.org/10.1016/j.physrep.2024.03.002},
   DOI={10.1016/j.physrep.2024.03.002},
   journal={Physics Reports},
   publisher={Elsevier BV},
   author={Blekos, Kostas and Brand, Dean and Ceschini, Andrea and Chou, Chiao-Hui and Li, Rui-Hao and Pandya, Komal and Summer, Alessandro},
   year={2024},
   month=jun, pages={1–66} }

@article{Abbas2024,
	title = {Challenges and opportunities in quantum optimization},
	volume = {6},
	issn = {2522-5820},
	url = {https://www.nature.com/articles/s42254-024-00770-9},
	doi = {10.1038/s42254-024-00770-9},
	number = {12},
	urldate = {2026-04-12},
	journal = {Nature Reviews Physics},
	author = {Abbas, Amira and Ambainis, Andris and Augustino, Brandon and Bärtschi, Andreas and Buhrman, Harry and Coffrin, Carleton and Cortiana, Giorgio and Dunjko, Vedran and Egger, Daniel J. and Elmegreen, Bruce G. and Franco, Nicola and Fratini, Filippo and Fuller, Bryce and Gacon, Julien and Gonciulea, Constantin and Gribling, Sander and Gupta, Swati and Hadfield, Stuart and Heese, Raoul and Kircher, Gerhard and Kleinert, Thomas and Koch, Thorsten and Korpas, Georgios and Lenk, Steve and Marecek, Jakub and Markov, Vanio and Mazzola, Guglielmo and Mensa, Stefano and Mohseni, Naeimeh and Nannicini, Giacomo and O’Meara, Corey and Tapia, Elena Peña and Pokutta, Sebastian and Proissl, Manuel and Rebentrost, Patrick and Sahin, Emre and Symons, Benjamin C. B. and Tornow, Sabine and Valls, Víctor and Woerner, Stefan and Wolf-Bauwens, Mira L. and Yard, Jon and Yarkoni, Sheir and Zechiel, Dirk and Zhuk, Sergiy and Zoufal, Christa},
	month = oct,
	year = {2024},
	pages = {718--735},
}

@article{Larocca2025,
  author  = {Mart{\'\i}n Larocca and Supanut Thanasilp and Samson Wang and Kunal Sharma and Jacob Biamonte and Patrick J. Coles and Lukasz Cincio and Jarrod R. McClean and Zo{\"e} Holmes and M. Cerezo},
  title   = {Barren plateaus in variational quantum computing},
  journal = {Nature Reviews Physics},
  year    = {2025},
  volume  = {7},
  number  = {4},
  pages   = {174--189},
  doi     = {10.1038/s42254-025-00813-9},
  url     = {https://doi.org/10.1038/s42254-025-00813-9},
  issn    = {2522-5820},
  month   = apr
}

@article{PhysRevLett.129.120501,
  title = {Coherence as a Resource for {Shor's} Algorithm},
  author = {Ahnefeld, Felix and Theurer, Thomas and Egloff, Dario and Matera, Juan Mauricio and Plenio, Martin B.},
  journal = {Phys. Rev. Lett.},
  volume = {129},
  issue = {12},
  pages = {120501},
  numpages = {7},
  year = {2022},
  month = {Sep},
  publisher = {American Physical Society},
  doi = {10.1103/PhysRevLett.129.120501},
  url = {https://link.aps.org/doi/10.1103/PhysRevLett.129.120501}
}

@article{PhysRevA.93.012111,
  title = {Coherence as a resource in decision problems: The {Deutsch--Jozsa} algorithm and a variation},
  author = {Hillery, Mark},
  journal = {Phys. Rev. A},
  volume = {93},
  issue = {1},
  pages = {012111},
  numpages = {6},
  year = {2016},
  month = {Jan},
  publisher = {American Physical Society},
  doi = {10.1103/PhysRevA.93.012111},
  url = {https://link.aps.org/doi/10.1103/PhysRevA.93.012111}
}

@article{PhysRevA.106.062429,
  title = {Entanglement and coherence in the {Bernstein--Vazirani} algorithm},
  author = {Naseri, Moein and Kondra, Tulja Varun and Goswami, Suchetana and Fellous-Asiani, Marco and Streltsov, Alexander},
  journal = {Phys. Rev. A},
  volume = {106},
  issue = {6},
  pages = {062429},
  numpages = {13},
  year = {2022},
  month = {Dec},
  publisher = {American Physical Society},
  doi = {10.1103/PhysRevA.106.062429},
  url = {https://link.aps.org/doi/10.1103/PhysRevA.106.062429}
}

@article{e21030260,
  author = {Liu, Ye-Chao and Shang, Jiangwei and Zhang, Xiangdong},
  title = {Coherence Depletion in Quantum Algorithms},
  journal = {Entropy},
  volume = {21},
  number = {3},
  pages = {260},
  year = {2019},
  doi = {10.3390/e21030260},
  url = {https://www.mdpi.com/1099-4300/21/3/260}
}

@article{FENG2023129048,
  title = {Coherence and entanglement in {Grover} and {Harrow--Hassidim--Lloyd} algorithm},
  author = {Feng, Changchun and Chen, Lin and Zhao, Li-Jun},
  journal = {Physica A: Statistical Mechanics and its Applications},
  volume = {626},
  pages = {129048},
  year = {2023},
  doi = {10.1016/j.physa.2023.129048},
  url = {https://www.sciencedirect.com/science/article/pii/S0378437123006039}
}

@article{haug2025pseudorandom,
  title = {Pseudorandom unitaries are neither real nor sparse nor noise-robust},
  author = {Haug, Tobias and Bharti, Kishor and Koh, Dax Enshan},
  journal = {Quantum},
  volume = {9},
  pages = {1759},
  year = {2025}
}

@article{miyazaki2022imaginarity,
  title = {Imaginarity-free quantum multiparameter estimation},
  author = {Miyazaki, Jisho and Matsumoto, Keiji},
  journal = {Quantum},
  volume = {6},
  pages = {665},
  year = {2022}
}

@article{ye2026coherence,
  title = {Coherence and imaginarity as resources in quantum circuit complexity},
  author = {Ye, Linlin and Wu, Zhaoqi and Zhou, Nanrun},
  journal = {Advanced Quantum Technologies},
  volume = {9},
  number = {4},
  pages = {e01007},
  year = {2026}
}

@article{s7kr-8rrn,
  title = {Quantum-imaginarity-based quantum speed limit},
  author = {Xuan, Dong-Ping and Shen, Zhong-Xi and Zhou, Wen and Fei, Shao-Ming and Wang, Zhi-Xi},
  journal = {Phys. Rev. A},
  volume = {112},
  issue = {5},
  pages = {052202},
  numpages = {19},
  year = {2025},
  month = {Nov},
  publisher = {American Physical Society},
  doi = {10.1103/s7kr-8rrn},
  url = {https://link.aps.org/doi/10.1103/s7kr-8rrn}
}

@article{Caliz2025,
  title = {A coherent approach to quantum-classical optimization},
  author = {C{\'a}liz, Andr{\'e}s N. and Riu, Jordi and Bosch, Josep and Torrente, Pau and Miralles, Jose and Riera, Arnau},
  journal = {Communications Physics},
  volume = {8},
  pages = {197},
  year = {2025},
  doi = {10.1038/s42005-025-02111-3}
}

@article{Sarmina2026,
  title = {Probing entanglement and parameter sensitivity in {QAOA} via Quantum Fisher Information},
  author = {Garc{\'i}a Sarmina, Brian and Saavedra Benavides, Jorge and Sun, Guo-Hua and Dong, Shi-Hai},
  journal = {Quantum Review Letters},
  volume = {2},
  pages = {1--20},
  year = {2026},
  doi = {10.1016/j.qrl.2025.12.001},
  url = {https://www.sciencedirect.com/science/article/pii/S3050491025000068}
}

\end{document}